\def\year{2018}\relax
\documentclass[letterpaper]{article} 
\usepackage{aaai18}  
\usepackage{times}  
\usepackage{helvet}  
\usepackage{courier}  
\usepackage{url}  
\usepackage{graphicx}  

\usepackage{amssymb,amsmath,latexsym,wasysym,dsfont,stmaryrd}
\usepackage{amsthm}
\usepackage{bm} 
\usepackage{dirtytalk} 
\usepackage{xparse} 
\usepackage{xspace}
\usepackage{useful-macros,environments,colored-comments}
\usepackage{tikz}

\title{Reasoning about Knowledge and Strategies under Hierarchical Information}
\author{Bastien Maubert \ and Aniello Murano\\
Universit\`a degli Studi di Napoli ``Federico II''}

\newcommand\UElogo{%
\begin{tikzpicture}[remember picture,overlay]
\node[anchor=south,yshift=1.5cm,xshift=0cm] at (current page.south) {\includegraphics[height=2.5em]{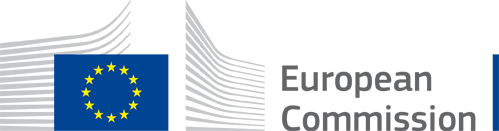}};
\end{tikzpicture}%
}

\begin{document}
\maketitle

\begin{abstract}
Two distinct semantics have been considered  for knowledge in the
context of strategic reasoning, depending on whether players know each
other's strategy or not. 
In the former case, that we call the \emph{informed} semantics,  distributed synthesis
for epistemic temporal specifications is  undecidable, already on systems with hierarchical
information. However, for the other, \emph{uninformed} semantics, the problem is
decidable on such systems. 
In this work we  generalise this result by introducing an
epistemic extension of Strategy Logic with imperfect information. The
semantics of knowledge operators is uninformed, and captures  agents that
can change  observation power when they change strategies. We solve the 
model-checking problem on a class of ``hierarchical instances'',
which provides a solution to a vast class of
strategic problems with epistemic temporal specifications, such as distributed or
rational synthesis, on
hierarchical systems. 
\end{abstract}

\section{Introduction}
\label{sec-intro}

Logics of programs, many of which  are based on  classic temporal logics
such as \LTL or \CTL, are meant to specify desirable properties of
algorithms.
When considering distributed systems, an important aspect is that
each processor has a partial, local view of the whole system, and can
only base its actions on the available information.
Many authors have argued that this imperfect information calls for
logical formalisms that would allow the modelling and reasoning about what different
processes know of the system, and of other processes'
state of knowledge. For instance, Halpern and Moses wrote
in~\cite{halpern1990knowledge}:

\say{[\ldots] explicitly reasoning about the states of knowledge of the
components of a distributed system provides a more general and uniform setting that
offers insight into the basic structure and limitations of protocols in a given system.}



To reason about knowledge and
time, temporal logics have been extended with the knowledge operator
$\K_a$ from epistemic logic, giving rise to a family of temporal
epistemic logics~\cite{fagin1995reasoning}, 
which have been applied 
to, \eg, information flow and cryptographic
protocols~\cite{DBLP:conf/csfw/MeydenS04,DBLP:journals/jcs/HalpernO05},
 coordination problems in distributed systems~\cite{neiger1992using} and
 motion planning in robotics~\cite{brafman1997applications}.

Distributed systems are often open systems, \ie, they interact with an environment and must react
appropriately to actions taken by this environment. 
As a result,
if we take the analogy
where processors are players of a game, and processes  are strategies
for the processors,
the task of synthesising distributed protocols can be seen as synthesising
winning strategies in multi-player games with imperfect information.
This analogy between the two settings is
well known, and Ladner and Reif already wrote
in~\cite{ladner1986logic} that \say{Distributed protocols are
  equivalent to (\ie, can be formally modelled as) games}. 

To reason about a certain type of game-related properties in
distributed systems, 
Alternating-time Temporal Logic
(\ATL) was introduced~\cite{DBLP:journals/jacm/AlurHK02}. It can express the
existence of strategies for coalitions of players in multi-player
games,
but cannot express some
important game-theoretic concepts, such as the existence of Nash
equilibria. To remedy this, Strategy Logic
(\SL)~\cite{DBLP:journals/iandc/ChatterjeeHP10,DBLP:journals/tocl/MogaveroMPV14}
was defined. Treating strategies as explicit first-order objects makes
it very expressive, and it can for instance talk about Nash
equilibria 
in a
very natural way.  These logics have been studied both for 
 perfect and imperfect information, and
in the latter case either with the assumption that agents have no
memory, or that they remember everything they observe. This last
assumption, called \emph{perfect recall}, is the one usually
considered in distributed
synthesis~\cite{PR90} 
and games with imperfect
information~\cite{reif84}, 
and it is also central in logics of knowledge and
time~\cite{fagin1995reasoning}. 
 It
is the one we consider in this work.

In order to reason about knowledge and strategic abilities in
distributed systems, epistemic temporal logics and strategic logics
have been combined. In particular, both \ATL and \SL have been
extended with knowledge
operators~\cite{Hoek03ATELstudialogica,DBLP:journals/fuin/JamrogaH04,epistemicSL,dima2010model,DBLP:conf/ijcai/BelardinelliLMR17,BelardinelliLMR17a}. However,
few decidable cases are known for the model checking of these logics
with imperfect information and perfect recall. This is not surprising
since strategic logics typically can express the existence of
distributed strategies, a problem known to be
undecidable for perfect recall, already for purely temporal
specifications~\cite{DBLP:conf/focs/PetersonR79,PR90}.  

 \head{Semantics of knowledge with strategies} 
Mixing knowledge and strategies raises intricate questions of
semantics. As a matter of fact, we find in the litterature two distinct semantics
for epistemic operators in strategic contexts, one in works 
on distributed synthesis from epistemic temporal specifications,  and
another one  in epistemic strategic logics. To explain
 in what they differ, let us first recall the semantics of the
 knowledge operator in epistemic temporal logics: a formula
 $\Ka\phi$ holds in a history $h$ (finite sequence of states) of a
 system if $\phi$ holds in all histories $h'$ that agent $a$ cannot
 distinguish from $h$.
 In other words, agent
 $a$ knows that $\phi$ holds if it holds in all histories  that  may
 be the current one according to what she observed. Now consider that the system is a
 multi-player game, and fix a strategy $\strat_b$ for some player $b$. 
 Two  semantics are possible for $\Ka$: one could say
 that $\Ka\phi$ holds if $\phi$ holds in
 all possible histories that are indistinguishable to the
 current one, as in epistemic temporal logics, or one could restrict attention to those in which player $b$
 follows $\strat_b$, discarding indistinguishable histories that are
 not consistent with $\strat_b$. 
In the
 latter case, one may say that player $a$'s knowledge is \emph{refined} by the knowledge of
 $\strat_b$, \ie, she knows that player $b$ is using strategy
 $\strat_b$, and she has the ability to refine her knowledge with this
 information, eliminating inconsistent possible histories. In the
 following this is what we will be referring to when saying that an
 agent \emph{knows} some other agent's strategy.
We shall also refer to the semantics where $a$ knows
 $\strat_b$ as the \emph{informed semantics}, and that in which she
 ignores it as the \emph{uninformed semantics}. 

 The two semantics are relevant as they model different reasonable
scenarios. For instance if players collaborate and have the ability to
communicate, they may share their strategies with each other.
But in many cases, components of a distributed system each receive
only their own strategy, and thus are ignorant of other
components' strategies.

All epistemic extensions of \ATL and \SL we know of consider the
uninformed
semantics.  
 In
contrast, works
on distributed synthesis from
epistemic temporal specifications use the informed
semantics~\cite{van1998synthesis,DBLP:conf/concur/MeydenW05}, even
though it is not so obvious that they do.
Indeed these works consider specifications in classic epistemic
temporal logic, without strategic operators. But they ask for the
existence of distributed  strategies so that  such specifications hold
in the system restricted to the outcomes of these strategies.
This corresponds to what we call the
informed semantics, as the semantics of knowledge operators is, in
effect, restricted to outcomes of the strategies that are being synthesised.

 We only know of two works that discuss these two
 semantics. In~\cite{DBLP:journals/iandc/BozzelliMP15} two
 knowledge-like operators are studied, one for each semantics, but
 distributed synthesis is not considered. More interestingly, Puchala
 already observes in~\cite{DBLP:conf/mfcs/Puchala10} that the
 distributed synthesis problem studied
 in~\cite{van1998synthesis,DBLP:conf/concur/MeydenW05} considers the
 informed semantics.  While the problem is undecidable for this
 semantics even on hierarchical
 systems~\cite{DBLP:conf/concur/MeydenW05}, Puchala sketches a proof that
 the problem becomes
 decidable on the class of hierarchical systems  when the uninformed
 semantics is used. 

\head{Contributions} We introduce a  logic for reasoning
about knowledge and strategies.  Our Epistemic Strategy Logic (\ESL) is based on 
Strategy Logic, and besides boolean and
temporal operators,  it contains the imperfect-information strategy quantifier $\Estrato{\obs}$
from \SLi~\cite{BMMRV17}, which reads as  ``there exists a strategy $x$ with
observation $\obs$'', and epistemic operators $\Ka$ for
each agent $\ag$. Our logic allows reasoning about  agents whose means of
observing the system changes over time, as agents may  successively use strategies
associated with different observations. This can model, for instance,
an agent that  is granted higher security clearance, giving her
access to previously hidden information.
The semantics of our epistemic operators
takes into account agents' changes of observational power. 
\ESL also contains the outcome quantifier $\A$
from Branching-time Strategy Logic (\BSL)~\cite{knight2015dealing},
which quantifies on outcomes of strategies currently used by the agents, and
the unbinding operator $\bind{\unbind}$, which frees an agent 
from her current strategy. The latter was introduced
in~\cite{DBLP:journals/iandc/LaroussinieM15} for \ATL with strategy
context and is also present in \BSL.
The outcome quantifier together with the unbinding operator allow us to express
branching-time temporal properties without resorting to artificial
strategy quantifications which may either affect the semantics of agents'
knowledge or break hierarchicality.

We solve the model-checking problem for
 \emph{hierarchical instances} of \ESL.
As in \SLi, hierarchical instances are formula/model pairs such that, as
one goes down the syntactic tree of the formula, observations
annotating strategy quantifiers $\Estrato{\obs}$ can only become
finer. In addition, in \ESL we require that knowledge operators do not  
 refer to (outcomes of) strategies quantified higher in the formula.

Any problem which can be expressed as
hierarchical instances of our logic is thus decidable, and since \ESL is 
very expressive such problems are many.
A first corollary is an alternative proof that
 distributed synthesis from epistemic
temporal specifications with uninformed semantics  is decidable on hierarchical systems. Puchala
 announced this result in~\cite{DBLP:conf/mfcs/Puchala10}, but we provide a
stronger result by going from linear-time to branching-time epistemic
specifications. We also allow for nesting of strategic operators in
epistemic ones, as long as hierarchicality is preserved and epistemic
formulas do not refer to previously quantified  strategies. 
Another corollary is  
that rational synthesis~\cite{DBLP:journals/amai/KupfermanPV16,DBLP:conf/icalp/ConduracheFGR16} with imperfect
information is decidable on hierarchical systems for epistemic temporal
objectives with uninformed semantics. 

Our approach to solve the model-checking problem for our logic extends
that followed in~\cite{DBLP:journals/iandc/LaroussinieM15,BMMRV17},
which consists in ``compiling'' the strategic logic under study into
an opportune variant of Quantified
\CTLs, or \QCTLs for
short~\cite{DBLP:journals/corr/LaroussinieM14}. This is an extension of \CTLs
with second-order quantification on propositions which serves as an
intermediary, low-level logic between strategic logics and tree
automata. In~\cite{DBLP:journals/iandc/LaroussinieM15}, model checking \ATLs
with strategy context is proved decidable by reduction to \QCTLs.
In ~\cite{BMMRV17}, model checking \SLi is proved decidable for a
class of
hierarchical instances by reduction to the hierarchical fragment of an imperfect information extension of
\QCTLs, called \QCTLsi. In this work we define \EQCTLsi, which extends further \QCTLsi with epistemic
operators and an operator of observation change introduced recently
in~\cite{barriere:2018} in the context of epistemic temporal
logics. 
We define the hierarchical fragment of \EQCTLsi, which strictly contains  that of
\QCTLsi, and solve its model-checking problem for this 
fragment. 

\vspace{1ex} \head{Related work} 
We know of five other logics called Epistemic Strategy Logic. 
In~\cite{huang2014temporal},  epistemic temporal logic is extended
with a first-order quantification
$\exists x$ on points in runs of the system and
an operator $\bm{e}_i(x)$ that  compares local state $i$  at $x$ and
at the current point. When interpreted on systems where strategies are
encoded in local states, this logic can express existence of
strategies and what agents know about it. However it only concerns
memoryless strategies. Strategy Logic is extended with epistemic
operators in~\cite{DBLP:conf/cav/CermakLMM14}, but they also consider
memoryless agents. \cite{epistemicSL}
extends a fragment of \SL 
with epistemic operators, and considers perfect-recall strategies, but
 model checking is not studied.  
The latter logic is extended in~\cite{DBLP:conf/ijcai/BelardinelliLMR17}, in
which its model-checking problem is solved on the class of broadcast
systems. 
In \cite{knight2015dealing}  \SL is also extended with epistemic operators and
 perfect-recall agents.
Their logic does not  require strategies to be uniform, but 
 this requirement can be expressed in the language. However no
 decidability result is provided. The result we present here is the
 first for an epistemic strategic logic with perfect recall on
 hierarchical systems. In addition, ours is the first  epistemic
 strategic logic to
 allow for changes of observational power.

\vspace{1ex}
\head{Plan} We first define \ESL and hierarchical instances, and
announce our main result. Next we introduce \EQCTLsi and
solve the
model-checking problem for its hierarchical fragment.
We then establish our main result by reducing model checking hierarchical instances
of \ESL to model checking hierarchical \EQCTLsi. We finally present two
corollaries, exemplifying what our logic can express, and we finish with a discussion
on the semantics of knowledge and strategies.

\subsection{Notations}
Let $\Sigma$ be an alphabet. A \emph{finite} (resp. \emph{infinite}) \emph{word} over $\Sigma$ is an element
of $\Sigma^{*}$ (resp. $\Sigma^{\omega}$).  
The \emph{length} of a finite word $w=w_{0}w_{1}\ldots
w_{n}$ is $|w|\egdef n+1$,  $\last(w)\egdef w_{n}$ is its last
letter, and we note $\eword$ for the empty word.
Given a finite (resp. infinite) word $w$ and $0 \leq i \leq |w|$ (resp. $i\in\setn$), we let $w_{i}$ be the
letter at position $i$ in $w$, $w_{\leq i}$ is the prefix of $w$ that
ends at position $i$ and $w_{\geq i}$ is the suffix of $w$ that starts
at position $i$.
We write $w\pref w'$ if $w$ is a prefix of $w'$, and $\FPref{w}$ is
the set of finite prefixes of word $w$.
Finally, 
the domain of a mapping $f$ is written $\dom(f)$,  and for $n\in\setn$ we let $[n]\egdef\{i \in \setn: 1 \leq i \leq n\}$.

We fix for  the rest of the paper a number of
parameters for our logics and models: $\APf$ is a finite set of
\emph{atomic propositions}, $\Agf$ is a finite set of \emph{agents} or
\emph{players},
$\Varf$ is a finite set of \emph{variables} and $\Obsf$ is a finite
set of \emph{observation symbols}.
These data are implicitly
  part of the input for the model-checking problems we consider.

  \section{Epistemic Strategy Logic}
\label{sec-ESL}

In this section we introduce our epistemic extension of Strategy Logic
with imperfect information.

\subsection{Models}
\label{sec-SLmodels}
 


The models of \ESL are essentially the same as those of \SLi, i.e., concurrent game structures
extended by an observation interpretation $\obsint$, that maps each observation
symbol 
$\obs\in\Obsf$ to an equivalence relation $\obsint(\obs)$ over positions of the
game structure.  
However models in \ESL  contain, in addition, an initial
observation for each player. This initial observation may change if the player receives a
strategy corresponding to a different observation.


\begin{definition}[\CGSi]
  \label{def-CGSi}
  A \emph{concurrent game structure with imperfect information} (or
  \CGSi for short)
is a structure  $\CGSi=(\Act,\setpos,\trans,\val,\obsint,\posinit,\obsinit)$ where
   \begin{itemize}
    \item $\Act$ is a finite non-empty set of \emph{actions},
    \item $\setpos$ is a finite non-empty set of \emph{positions},
   \item $\trans:\setpos\times \Mov^{\Agf}\to \setpos$ is a \emph{transition function}, 
  \item $\val:\setpos\to 2^{\APf}$ is a \emph{labelling function},  
  \item $\obsint:\Obsf\to \setpos\times\setpos$ is an 
  \emph{observation interpretation}, and for each $\obs\in\Obsf$, 
  $\obsint(\obs)$ is an equivalence relation,
  \item $\posinit \in \setpos$ is an \emph{initial position}, and
 \item $\obsinit\in \Obsf^\Agf$ is a tuple of \emph{initial observations}. 
\end{itemize}
\end{definition}

Two positions being equivalent for relation
$\obsint(\obs)$ means that a player
using a strategy with observation $\obs$ cannot distinguish them.
In the following we may  write $\obseq$ for $\obsint(\obs)$ and $\pos\in\CGS$ for $\pos\in\setpos$.



\vspace{1ex}
\head{Joint actions}
When in a position $\pos\in\setpos$, each player $\ag$ chooses an action $\mova\in\Mov$
and the game proceeds to position
$\trans(\pos, \jmov)$, where $\jmov\in \Mov^{\Agf}$ stands for the \emph{joint action}
$(\mova)_{\ag\in\Agf}$. If
$\jmov=(\mova)_{\ag\in\Agf}$, we let
$\jmov_{\ag}$ denote $\mova$ for $\ag\in\Agf$. 

\vspace{1ex}
\head{Plays and strategies}
A \emph{finite} (resp. \emph{infinite}) \emph{play} is a finite (resp. infinite)
word $\fplay=\pos_{0}\ldots \pos_{n}$ (resp. $\iplay=\pos_{0} \pos_{1}\ldots$)
such that $\pos_{0}=\pos_{\init}$ and for all $i$ with $0\leq i<|\fplay|-1$ (resp. $i\geq 0$), there exists a joint action $\jmov$
such that $\trans(\pos_{i}, \jmov)=\pos_{i+1}$. 
We let $\setplays$ be the set of finite plays. 
A \emph{strategy} is a function $\strat:\setplays \to \Mov$, and we let
$\setstrat$ be the set of all strategies.

\vspace{1ex}
\head{Assignments}
An \emph{assignment}  $\assign:\Agf\cup\Varf \partialto \setstrat$ is
a partial function assigning to
each  player and variable in its domain a strategy.
For an assignment
$\assign$, a player $\ag$ and a strategy $\strat$,
$\assign[\ag\mapsto\strat]$ is the assignment of domain
$\dom(\assign)\cup\{\ag\}$ that maps $\ag$ to $\strat$ and is equal to
$\assign$ on the rest of its domain, and 
$\assign[\var\mapsto \strat]$ is defined similarly, where $\var$ is a
variable; also, $\assign[\ag\mapsto\unbind]$ is
 the assignment of domain $\dom(\assign)\setminus\{\ag\}$, on which it is
 equal to $\assign$.

 \vspace{1ex}
 \head{Outcomes}
For an assignment $\assign$ and a finite play $\fplay$, we let
$\out(\assign,\fplay)$ be the set of infinite plays that start with
$\fplay$ and are then extended by letting players follow the strategies
assigned by $\assign$. Formally,
 $\out(\assign,\fplay)$ is the set of infinite plays of the form $\fplay \cdot
 \pos_{1}\pos_{2}\ldots$ such that for all $i\geq 0$, there exists
 $\jmov$ such that for all $\ag\in\dom(\assign)\cap\Agf$,
 $\jmov_\ag\in\assign(\ag)(\fplay\cdot\pos_{1}\ldots\pos_i)$ and 
 $\pos_{i+1}=\trans(\pos_{i},\jmov)$, with 
 $\pos_{0}=\last(\fplay)$.



 \vspace{1ex}
\head{Synchronous perfect recall}
Players with perfect recall 
remember the whole history of a play. 
Each observation
 relation is thus extended to finite plays
as follows: $\fplay \obseqh \fplay'$ if $|\fplay|=|\fplay'|$
and $\fplay_{i}\obseq\fplay'_{i}$ for every $i\in\{0,\ldots,
|\fplay|-1\}$.

\vspace{1ex}
\head{Uniform strategies} For $\obs\in\Obsf$, an \emph{$\obs$-strategy} is a strategy $\strat:\setpos^{+} \to \Mov$ 
such that $\strat(\fplay)=\strat(\fplay')$ whenever $\fplay \obseqh
\fplay'$. 
For $\obs\in\Obsf$ we let $\setstrato$ be the set of
all $\obs$-strategies. 

\subsection{Syntax}
\label{sec-ESL-definition}

 \ESL   extends \SLi with knowledge
  operators $\Ka$     
  for each agent $\ag\in\Agf$, and the 
 \emph{outcome quantifier} from
Branching-time Strategy Logic, introduced in~\cite{knight2015dealing},
which quantifies on outcomes of the currently fixed strategies.
While in \SL  temporal operators
could only be evaluated in contexts where all agents were assigned to a
strategy, this outcome quantifier allows for evaluation of
(branching-time) temporal properties on partial assignments of
strategies to agents. This outcome quantifier can be simulated in
usual, linear-time variants of Strategy Logic, by quantifying on
strategies for agents who do not currently have one. But in the context of imperfect information, where strategy
quantifiers are parameterised by an observation, this may cause to
either break the hierarchy or artificially modify an agent's
observation, which affects his knowledge.


\begin{definition}[\ESL Syntax]
  \label{def-ESL}
  The syntax of \ESL is defined by the following grammar:
    \begin{align*}
 \phi &: p \mid  \neg\phi \mid \phi\ou\phi   \mid
\Estrato{\obs} \phi \mid \bind{\var}\phi \mid \bind{\unbind}\phi  \mid \Ka \phi \mid \Aout \psi\\
\psi &: \phi \mid \neg \psi \mid
  \psi \ou \psi \mid \X \psi \mid \psi \until \psi, \end{align*}
where $p\in\APf$, $\var\in\Varf$, $\obs\in\Obsf$ and $a\in\Agf$. 
\end{definition}

Formulas of type $\phi$ are called \emph{history formulas}, those of
type $\psi$ are \emph{path formulas}. 
We may use usual  abbreviations $\top\egdef p\ou\neg p$, $\perp\egdef\neg\top$, $\phi\impl\phi'\egdef \neg \phi \ou \phi'$,
$\F\phi \egdef \top \until \phi$,  $\G\phi \egdef \neg \F
\neg \phi$,  $\Astrato{\obs}\phi\egdef\neg\Estrato{\obs}\neg\phi$ and $\E\psi\egdef \neg\A\neg\psi$.

A variable $\var$ appears \emph{free} in a formula $\phi$ if it
appears out of the scope of a strategy quantifier. 
We let $\free(\phi)$ be the set of variables that appear
free in $\phi$.
An assignment $\assign$ is \emph{variable-complete} for a formula
$\phi$ if its domain contains all free variables in $\phi$.
Finally, a \emph{sentence} is a history formula $\phi$ such that $\free(\phi)=\emptyset$.
\begin{remark}
  \label{rem-obs-strat}
  Without loss of generality we can assume that each strategy variable
$\var$ is quantified at most once in an \ESL formula. Thus, each
variable $\var$ that appears in a
sentence is uniquely associated 
to a strategy quantification $\Estrato{\obs}$,
and we let $\obs_\var=\obs$.
\end{remark}

\vspace{1ex} \head{Discussion on the syntax} In \SLi as
  well as in \ESL, the observation used by a strategy is specified at
  the level of strategy quantification: $\Estrato{\obs}\phi$ reads as
  ``there exists a strategy $\var$ with observation $\obs$ such that
  $\phi$ holds''.  When a strategy with observation $\obs$ is assigned
  to some agent $\ag$ via a binding $(\ag,\var)$, it seems natural to
  consider that agent $\ag$ starts observing the system with
  observation $\obs$.  As a result agents can change observation when
  they change strategy, and thus they can observe the game with
  different observation powers along a same play. This
  contrasts with most of the literature on epistemic temporal logics,
  where each agent's observation power is usually fixed in the model. 


 \subsection{Epistemic relations with changing observations}
 \label{sec-perf-rec}

Dynamic observation change has been studied recently in the context of
epistemic temporal logics in~\cite{barriere:2018}, from which come the
following definitions. 

First, dealing with the possibility to dynamically change  observation  requires to remember which observation each agent had at
each point in time. 

\vspace{1ex}
\head{Observation records} An \emph{observation record} $\reca$ for
agent $a\in\Agf$ is a finite word over $\setn\times\Obsf$, \ie,
$\reca\in (\setn\times\Obsf)^*$.

If at time $n$, an agent $\ag$ with
current observation record $\reca$ receives a strategy with
observation $\obs$, her new observation record is
$\append{(\obs,n)}$. The observation record $\rn{n}$ is the projection of
$\reca$ on $\{n\}\times\Obsf$, and represents the sequence of
observation changes that occurred at time $n$.

Given  an observation record for each agent  $\rec=(\reca)_{a\in\Agf}$,
 we note $\rec_a$ for $\reca$.
We say that an observation record $\rec$ \emph{stops at time $n$} if $\rn{m}$
is empty for all $m>n$, and $\rec$ \emph{stops at a finite play
  $\fplay$} if it stops at time $|\fplay|-1$.
If $\rec$ stops at time $n$, we let $\recadd{n}{\obs}$ be the
 observation record $\rec$ where $\rec_a$ is replaced with
 $\rec_a\cdot (n,\obs)$.

At each step of a play, each agent observes the new position with her
current observation power. Then, if an agent changes strategy, she
observes the same position with the observation of the new strategy,
which may be different from the previous one. 
Also, due to the syntax of \ESL,
an agent may change observation  several times before the next step.
Therefore, the 
 \emph{observation sequence} $\obslista(\rec,n)$ with which  agent $\ag$ observes the game at time $n$ consists of the
observation she had when the $n$-th move is taken,
plus those corresponding to strategy changes that occur before the
next step. It is defined by
induction on $n$:
\[
\begin{array}{l}
\obslista(\rec,0)=   o_1 \cdot \ldots \cdot o_k,\\[1pt]
 \hfill \mbox{if }  \rvecn{0}=(0,\obs_1)\cdot\ldots\cdot(0,\obs_k), \mbox{ and}\\[2pt]
 \obslista(\rec,n+1)= \last(\obslista(\rec,n))\cdot o_1 \cdot \ldots \cdot o_k,\hspace{1.5cm}\\[1pt]
\hfill \mbox{if } \rvecn{n+1}=(n+1,\obs_1) \cdot \ldots \cdot (n+1,\obs_k).
\end{array}
\]

If at time $n$ agent $a$ does not receive a new strategy, $\obslista(\rec,n)$ contains only one observation,
which will be either that of the last strategy
taken by the agent or the agent's  initial observation, given by the \CGSi. 

The  indistinguishability
relation for synchronous perfect-recall  with observation change is defined as follows. 

\begin{definition}
  \label{def-ch-relation}
  For $\fplay$ and $\fplay'$ two finite plays and $\rec$
an observation record, $\fplay$ and $\fplay'$ are observationally
equivalent to agent $a$, written $\fplay \eqha{\rec} \fplay'$, if $|\fplay|=|\fplay'|$
and,  for every $i\in\{0,\ldots,
|\fplay|-1\}$, for every $\obs\in\obslista(\rec,i)$, $\fplay_{i}\obseq[\obs]\fplay'_{i}$.
\end{definition}

\begin{remark}
  \label{rem-order-obs}
Observe that, at a given point in time, the order in which an agent
observes the game with different observation does not
matter. Intuitively, all that matters is the total 
information gathered before the next step. Also, in the case of an empty observation record,
the above definition corresponds to blind agents, for which all finite
plays of same length are indistinguishable. However in the following
observation records will never be empty, but will always be initialised with
the initial observations given by the model.
\end{remark}

\subsection{Semantics}
\label{sec-semantics-ESL}

We now define the semantics of \ESL. 

\begin{definition}[\ESL Semantics]
  \label{def-ESL-semantics}
  The semantics of a history formula is defined on a game $\CGSi$
  (omitted below), an
assignment  $\assign$  variable-complete for $\varphi$, and a
finite play $\fplay$. For a path formula $\psi$, the finite play is
replaced with an infinite play $\iplay$ and an index $i\in\setn$. The
definition  is as follows: 
\[
\begin{array}{lcl}
 \assign,\rec,\fplay\modelsSL p & \text{if} & p\in\val(\last(\fplay))\\[1pt]
 \assign,\rec,\fplay\modelsSL \neg\varphi & \text{if} &
  \assign,\rec,\fplay\not\modelsSL\varphi\\[1pt]
 \assign,\rec,\fplay\modelsSL \varphi\vee\varphi' & \text{if} &
  \assign,\rec,\fplay\modelsSL\varphi \;\text{ or }\;
  \assign,\rec,\fplay\modelsSL\varphi' \\[1pt]
 \assign,\rec,\fplay\modelsSL\Estrato{\obs} \varphi  & \text{if} & 
\exists   \strat\in\setstrato. \;
    \assign[\var\mapsto\strat],\rec,\fplay\modelsSL \varphi\\[1pt]
 \assign,\rec,\fplay\modelsSL \bind{\var}\varphi & \text{if} &
 \assign[\ag\mapsto\assign(\var)],\rec,\fplay\modelsSL \varphi\\[1pt]  
 \assign,\rec,\fplay\modelsSL \bind{\unbind}\varphi & \text{if} &
                                                                    \assign[\ag\mapsto\unbind],\rec,\fplay\modelsSL \varphi\\[1pt]
\assign,\rec,\fplay\modelsSL \Ka\phi & \text{if} &
                                                   \forall\fplay'\in\setplays \text{ s.t. }\fplay'\eqha{\rec}\fplay,\\
  & & \assign,\rec,\fplay'\modelsSL\phi \\[1pt]

  \assign,\rec,\fplay\modelsSL \Aout\psi & \text{if} & \forall\iplay
                                                       \in
                                                       \out(\assign,\fplay),\\
  & &\assign,\rec,\iplay,|\fplay|-1\modelsSL  \psi \\[5pt]
\assign,\rec,\iplay,i\modelsSL \varphi & \text{if} &
                                                         \assign,\rec,\iplay_{\leq i}\modelsSL\varphi\\[1pt]
   \assign,\rec,\iplay,i\modelsSL \neg\psi & \text{if} &
  \assign,\rec,\iplay,i\not\modelsSL\psi\\[1pt]
 \assign,\rec,\iplay,i\modelsSL \psi\vee\psi' & \text{if} &
  \assign,\rec,\iplay,i\modelsSL\psi \;\text{ or }\;
  \assign,\rec,\iplay,i\modelsSL\psi' \\[1pt]
  \assign,\rec,\iplay,i\modelsSL\X\psi & \text{if} &
  \assign,\rec,\iplay,i+1\modelsSL\psi\\[1pt]
\assign,\rec,\iplay,i\modelsSL\psi\until\psi' & \text{if} & \exists\, j\geq i
   \mbox{ s.t. }\assign,\rec,\iplay,j\modelsSL \psi' \mbox{ and }   \\[1pt]
& &   \forall\, k \in [i,j[,
\; \assign,\rec,\iplay,k\modelsSL \psi
\end{array}
\]
\end{definition}

The satisfaction of a sentence is independent of the
assignment; for an \ESL sentence $\phi$ we thus let
$\CGSi,\rec,\fplay\models\phi$ if $\CGSi,\assign,\rec,\fplay\models\phi$ for some
assignment $\assign$. We also write $\CGSi \models \phi$ if
$\CGSi,\eobsrecvec, \posinit \models \phi$, where $\eobsrecvec=(0,\obsinit_\ag)_{\ag\in\Agf}$.


\vspace{1ex}
\bmchanged{\head{Discussion on the semantics}
  First, 
 the semantics of the knowledge
operator corresponds, as announced, to what we called \emph{uninformed semantics} in
the introduction. Indeed it is not restricted to outcomes of strategies
followed by the players: $\Ka\phi$ holds in a finite play $\fplay$ if $\phi$
holds in \emph{all finite plays in the game} that are indistinguishable to
$\fplay$ for agent $a$. 

Also, note that the relation for perfect recall with observation
change, and thus also observation records, are only used in the semantics of the knowledge operators. As
usual, a strategy with observation $\obs$ has to be uniform with
regards to the classic perfect-recall relation $\obseq$ for static
observations, even if it is assigned to an agent who previously had a
different observation. The reasons to do so are twofold.

First, we do not see any other natural
definition. One may think of parameterising strategy quantifiers with
observation records instead of mere observations, but this would
require to know at the level of quantification at which points in time
the strategy will be given to an agent, and what previous observations
this agent would have had, which is not realistic.

More importantly,
when one asks for the existence of a uniform strategy after some
finite play $\fplay$, it only matters how the strategy is defined on
suffixes of $\fplay$, and thus the uniformity constraint also is
relevant only on such plays.
But for such plays, the ``fixed observation'' indistinguishability
relation is the same as the ``dynamic observation'' one. 
More precisely,  if agent $\ag$
receives  observation $\obs$ at the end of $\fplay$, \ie,
$\last(\rec_a)=(|\fplay|-1,\obs)$, then for all finite plays
$\fplay',\fplay''$ that are suffixes of $\fplay$, we have
$\fplay'\obseqh\fplay''$ if, and only if,
$\fplay'\eqha{\rec}\fplay''$. Indeed, since we use the S5 semantics of
knowledge, \ie, indistinguishability relations are equivalence
relations, the prefix $\fplay$ is always related to itself, be it for
$\eqha{\rec}$ or $\obseqh$, and after $\fplay$ both relations only
consider observation $\obs$.
}

\subsection{Model checking and hierarchical instances}


We now introduce the decision problem studied in
this paper, i.e., the model-checking problem for \ESL.

\vspace{1ex}
\head{Model checking}   An
\emph{instance} is a pair $(\Phi,\CGSi)$ where
$\Phi$ is a sentence of $\ESL$  and $\CGSi$ is a \CGSi.  The \emph{model-checking
  problem} for \ESL is the decision problem that, given an
instance $(\Phi,\CGSi)$, returns `yes' if $\CGSi\models\Phi$, and
`no' otherwise.

\SLi can be translated into \ESL, by adding outcome quantifiers
before temporal operators. Since model checking \SLi is undecidable~\cite{BMMRV17}, we get the following result:

\begin{theorem}
  \label{theo-undecidable-SLi}
  Model checking  \ESL is undecidable.
\end{theorem}

\head{Hierarchical instances} 
We now isolate a sub-problem obtained by restricting
attention to \emph{hierarchical instances}.
Intuitively, an \ESL-instance $(\Phi,\CGSi)$ is hierarchical if, as one goes down a path in the syntactic tree of $\Phi$, 
the observations parameterising strategy quantifiers become finer. In
addition,  epistemic formulas must not talk about
 currently defined strategies.

 Given an \ESL sentence $\Phi$ and a syntactic subformula $\phi$ of
 $\Phi$, by parsing $\Phi$'s  syntactic tree one can define the set $\Agphi{\phi}$ of agents who
 are  bound to a strategy at the level of $\phi$, as well as where in $\Phi$
 these
 strategies are quantified upon.
\begin{definition}
  \label{def-unconstrained}
Let $\Phi$ be an \ESL sentence. A subformula $\Ka\phi$  is
  \emph{free} if for every 
subformula $\A\psi$ of $\phi$, the current strategy of each agent in
$\Agphi{\A\psi}$ is quantified within $\phi$.
\end{definition}

In other words, an epistemic subformula $\Ka\phi$ is free if it 
does not  talk about strategies that are quantified before it.

 \begin{example}
   \label{ex-free}
   If $\Phi=\Estrato{\obs}\bind{\var}\Ka\Aout\X p$,
   then $\Ka\Aout\X p$ is not free in $\Phi$, because at the level of
   $\Aout$ agent $a$ is bound to strategy $\var$ which is
   quantified ``outside'' of $\Ka\Aout\X p$.
But   if $\Phi=\Estrato{\obs}\bind{\var}\Ka\bind{\unbind}\Aout\X p$, then
   $\Ka\bind{\unbind}\Aout\X p$ is free in $\Phi$, because at the
   level of $\A$ no agent is bound to a strategy.
Also if
$\Phi=\Estrato{\obs}\bind{\var}\Ka\Estrato[\varb]{\obs'}\bind{\varb}\Aout\X
p$, then $\Ka\Estrato[\varb]{\obs'}\bind{\varb}\Aout\X
p$ is free in $\Phi$, because at the level of $\Aout$ the only agent
bound to a strategy is $a$, and her strategy is quantified upon after the
knowledge operator.
\end{example}

We can now define the
hierarchical fragment for which we establish decidability of the
model-checking problem.

\begin{definition}[Hierarchical instances]
  \label{def-hierarchical}
  An \ESL-instance $(\Phi,\CGSi)$ is \emph{hierarchical} if all
  epistemic subformulas of $\Phi$ are free in $\Phi$ and, for all subformulas 
  of the form $\phi_1=\Estrato{\obs_1}\phi'_1$ and
  $\phi_2=\Estrato{\obs_2}\phi'_2$ where $\phi_2$ is a subformula of
  $\phi'_1$, it holds that $\obsint(\obs_2)\subseteq\obsint(\obs_1)$.
\end{definition}


In other words, an instance is hierarchical if innermost strategy quantifiers
 observe at least as much as
outermost ones, and epistemic formulas do not talk about current strategies.
Here is the main contribution of this work:


\begin{theorem}
\label{theo-ESL}
The model-checking problem for \ESL restricted to the class of hierarchical instances is decidable.
\end{theorem}

We  prove this result by reducing it to the model-checking 
problem for the hierarchical fragment of an extension of  \QCTLs with
imperfect information, knowledge and observation change, which we now introduce and
study in order to use it as an
intermediate, ``low-level'' logic between tree automata and \ESL.

\section{\QCTLsi with knowledge and observation change}
\label{sec-mQCTLi}

\QCTLs  extends \CTLs with second order
quantification on atomic propositions~\cite{emerson1984deciding,Kup95,KMTV00,french2001decidability,DBLP:journals/corr/LaroussinieM14}. It was recently extended to model imperfect-information aspects, resulting in the logic
called \QCTLsi~\cite{BMMRV17}.
In this section 
we first define an epistemic extension of \QCTLsi with operators for
knowledge and dynamic observation change, that we call \EQCTLsi.
Then we define
the syntactic class of \emph{hierarchical formulas} and prove that model checking this class of formulas is decidable.

\subsection{Models}
\label{sec-models-EQCTLi}

The models of \EQCTLsi, as those of \QCTLsi, are structures in which
states are tuples of local states. Fix $n\in\setn$.

\vspace{1ex}
\head{Local states}
Let $\{\setlstates_{i}\}_{i\in [n]}$ denote $n$ disjoint finite sets of
\emph{local states}. For  $I\subseteq [n]$, we let $\Dirtreei\egdef\prod_{i\in
I}\setlstates_{i}$ if $I\neq\emptyset$, and
$\Dirtreei[\emptyset]\egdef\{\blank\}$ where $\blank$ is a special symbol.

\vspace{1ex}
\head{Concrete observations}
A set $\cobs \subseteq [n]$ is a \emph{concrete observation} (to distinguish
 from observation symbols $\obs$ of \ESL).

Fix $\cobs \subseteq [n]$ and $I \subseteq [n]$.
 Two tuples $\dir,\dir'\in\Dirtreei[I]$ are
 \emph{$\cobs$-indistinguishable}, written  
  $\dir\oequiv\dir'$, if
 for each $i\in I\cap\,\cobs$, $\dir_i=\dir'_i$.
 Two words
   $\noeud=\noeud_{0}\ldots\noeud_{i}$ and
   $\noeud'=\noeud'_{0}\ldots\noeud'_{j}$ over alphabet $\Dirtreei[I]$ are
   \emph{$\cobs$-indistinguishable}, written $\noeud\oequivt\noeud'$, if
   $i=j$ and for all $k\in \{0,\ldots,i\}$ we have
   $\noeud_{k}\oequiv\noeud'_{k}$.

 \vspace{1ex}
\head{Compound Kripke structures} 
These are like Kripke structures
except that the states are elements of $\Dirtreei[{[n]}]$.
A \emph{compound Kripke structure}, or \CKS, over $\APf$, is a tuple 
$\CKS=(\setstates,\relation,\lab,\stateinit,\cobsvecinit)$ where
$\setstates\subseteq \Dirtreei[{[n]}]$ is a set of \emph{states}, 
$\relation\subseteq\setstates\times\setstates$ is a
left-total\footnote{\ie, for all $\state\in\setstates$, there exists $\state'$
such that $(\state,\state')\in\relation$.} \emph{transition relation}, 
$\lab:\setstates\to 2^{\APf}$ is a \emph{\labeling function},
$\stateinit \in \setstates$ is an \emph{initial state}, and
$\cobsvecinit=(\cobsinit_\ag)_{\ag\in\Agf}$ is an \emph{initial
  concrete observation} for each agent.

A \emph{path} in $\CKS$  is an infinite sequence of states
$\spath=\state_{0}\state_{1}\ldots$ such that
 for all $i\in\setn$,
$(\state_{i},\state_{i+1})\in \relation$, and a \emph{finite path}
$\fspath=\state_{0}\state_{1}\ldots\state_n$ is a finite prefix of a path. For 
$\state\in\setstates$, we let $\Paths(\state)$ be the set of all
paths that start in $\state$, and $\FPaths(\state)$ is the set of
finite paths that start in $\state$.  

\subsection{Syntax of \EQCTLsi}
\label{sec-syntax-EQCTLi}

The syntax of \EQCTLsi extends that of \QCTLsi with epistemic
operators $\Ka$ and observation-change operators $\Da{\cobs}$,
which were recently introduced and studied in~\cite{barriere:2018} in an epistemic temporal
logic without second-order quantification.

\begin{definition}[\EQCTLsi Syntax]
  \label{def-syntax-EQCTLsi}
  The syntax of \EQCTLsi is defined by the following grammar:
  \begin{align*}
  \phi\egdef &\; p \mid \neg \phi \mid \phi\ou \phi \mid \A \psi \mid
  \existsp[p]{\cobs} \phi \mid \Ka \phi \mid \Da{\cobs}\phi\\
    \psi\egdef &\; \phi \mid \neg \psi \mid \psi\ou \psi \mid \X \psi \mid
  \psi \until \psi
\end{align*}
where $p\in\APf$, $a\in\Agf$ 
and $\cobs\subseteq [n]$.
\end{definition}

Formulas of type $\phi$ are called \emph{state formulas}, those of type $\psi$
are called \emph{path formulas}, and \EQCTLsi consists of all the state formulas.
$\A$ is the classic path quantifier from branching-time temporal
logics.
$\exists^\cobs$ is the second-order quantifier with imperfect
information from \QCTLsi~\cite{BMMRV17}. $\existsp[p]{\cobs}\phi$
holds in a tree if there is way to choose a labelling for $p$ such that $\phi$
holds, with the constraint that $\cobs$-equivalent nodes of the tree
must be labelled identically.
$\Ka\phi$  means ``agent $\ag$ knows that $\phi$ holds'',
where the knowledge depends on the sequence of observations agent
$\ag$ has had; finally, $\Da{\cobs}\phi$ means that after agent $\ag$ switches to
observation $\cobs$, $\phi$ holds.

Given an \EQCTLsi formula $\phi$, we define the set of
\emph{quantified propositions} $\APq(\phi)\subseteq\APf$ as the set of
atomic propositions $p$ such that $\phi$ has a subformula of the form
$\existsp[p]{\cobs}\phi$. We also define the set of \emph{free
  propositions} $\APfree(\phi)\subseteq\APf$ as the set of atomic
propositions $p$ that appear out of the scope of any quantifier of the
form $\existsp[p]{\cobs}$

\subsection{Semantics of \EQCTLsi}
\label{sec-EQCTLsi-semantics}

Before defining the semantics of the logic we first recall some
definitions for trees.

\vspace{1ex}
\head{Trees}
Let $\Dirtree$ be a finite set (typically a set of states). 
An \emph{$\Dirtree$-tree} $\tree$ 
 is a
nonempty set of words $\tree\subseteq \Dirtree^+$ such that:
\begin{itemize}
  \item\label{p-root} there exists $\racine\in\Dirtree$,  called the
    \emph{root} of $\tree$, such that each
    $\noeud\in\tree$ starts with $\racine$ (\ie, $\racine\pref\noeud$);
  \item if $\noeud\cdot\dir\in\tree$ and $\noeud\neq\epsilon$, then
    $\noeud\in\tree$, and
  \item if $\noeud\in\tree$ then there exists $\dir\in\Dirtree$ such that $\noeud\cdot\dir\in\tree$.
\end{itemize}

The elements of a tree $\tree$ are called \emph{nodes}.
  If 
 $\noeud\cdot\dir \in \tree$, we say that $\noeud\cdot\dir$ is a \emph{child} of
 $\noeud$.
A \emph{\tpath} in $\tree$ is an infinite sequence of nodes $\tpath=\noeud_0\noeud_1\ldots$
such that for all $i\in\setn$, $\noeud_{i+1}$ is a child of
$\noeud_i$,
and $\tPaths(\noeud)$ is the set of \tpaths
 that start in node $\noeud$.  
An \emph{$\APf$-\labeled $\Dirtree$-tree}, or
\emph{$(\APf,\Dirtree)$-tree} for short, is a pair
$\ltree=(\tree,\lab)$, where $\tree$ is an $\Dirtree$-tree called the
\emph{domain} of $\ltree$ and
$\lab:\tree \rightarrow 2^{\APf}$ is a \emph{\labeling}.
For a labelled tree $\ltree=(\tree,\lab)$ and an atomic
proposition $p\in\APf$, we define the \emph{$p$-projection of $\ltree$}
as the labelled tree
$\proj{\ltree}\;\egdef(\tree,\proj{\lab})$, where for each
$\noeud\in\tree$, $\proj{\lab}\!(\noeud)\egdef \lab(\noeud)\setminus
\{p\}$. 
Two labelled trees $\ltree=(\tree,\lab)$ and
  $\ltree'=(\tree',\lab')$ are \emph{equivalent modulo $p$},
  written $\ltree\Pequiv\ltree'$, if
  $\proj{\ltree}=\proj{\ltree'}$ (in particular,
  $\tree=\tree'$).

\vspace{1ex}
\head{Quantification and uniformity} In \EQCTLsi, as in \QCTLsi,
$\existsp[p]{\cobs} \phi$ holds in a tree $\ltree$ if there is some
$\cobs$-uniform $p$-labelling of $\ltree$ such that $\ltree$ with this
$p$-labelling satisfies $\phi$.  A $p$-labelling of a
tree is $\cobs$-uniform  if 
 every two nodes that are indistinguishable for observation $\cobs$
agree on
their $p$-labelling.


\begin{definition}[$\cobs$-uniformity]
  \label{def-uniformity}
 A labelled tree $\ltree=(\tree,\lab)$ is \emph{$\cobs$-uniform in $p$} if for every pair of nodes
 $\noeud,\noeud'\in\tree$ such that $\noeud\oequivt\noeud'$, we have
 $p\in\lab(\noeud)$ iff $p\in\lab(\noeud')$.  
\end{definition}

\head{Changing observations} To capture how the observation-change operator affects
the semantics of the knowledge operator, we use again observation
records $\rec$ and the associated notion of observation sequence
$\obslista(\rec,n)$. They are defined as for \ESL
except that observation symbols $\obs$ are replaced with concrete
observations $\cobs$.
  For $\noeud=\noeud_{0}\ldots\noeud_{i}$ and
  $\noeud'=\noeud'_{0}\ldots\noeud'_{j}$ over alphabet $\Dirtreei[I]$,
  and 
an observation record  $\rec$, we say that $\noeud$ and $\noeud'$ are observationally
equivalent to agent $a$, written $\noeud \eqha{\rec} \noeud'$, if $i=j$
and,  for every $k\in\{0,\ldots,
i\}$ and every $\cobs\in\obslista(\rec,k)$, $\noeud_{k}\oequiv[\cobs]\noeud'_{k}$.

Finally, we inductively define the satisfaction relation $\models$. Let   $\ltree=(\tree,\lab)$ be
a $2^{\APf}$-\labeled $\Dirtreei$-tree, 
$\noeud$  a node and $\rec$ an observation record that stops at $\noeud$:
\[
  \begin{array}{lcl}
  \ltree,\rec,\noeud\models p 			& \mbox{ if } & p\in\lab(\noeud)\\
  \ltree,\rec,\noeud\models\neg \phi		& \mbox{ if } & \ltree,\rec,\noeud\not\models \phi\\
  \ltree,\rec,\noeud\models \phi \ou \phi'		& \mbox{ if } &
                                                                \ltree,\rec,\noeud \models \phi \mbox{ or }\ltree,\rec,\noeud\models \phi' \\
  \ltree,\rec,\noeud\models \E\psi	& \mbox{ if } & \exists\,\tpath\in\tPaths(\noeud) \mbox{ s.t. }\ltree,\rec,\tpath\models \psi \\
  \ltree,\rec,\noeud\models \existsp{\cobs} \phi & \mbox{ if }
  & \exists\,\ltree'\Pequiv[p]\ltree \mbox{ s.t.
  }\ltree'\mbox{ is $\cobs$-uniform in $p$}\\
&    & \mbox{and }\ltree',\rec,\noeud\models\phi\\
    \ltree,\rec,\noeud\models \Ka \phi & \mbox{ if }
  & \forall \noeud'\in\ltree \mbox{ s.t. }\noeud\eqha{\rec}\noeud',\\
& & \ltree,\rec,\noeud'\models\phi
  \end{array}\]
And if $\tpath$ is a path in $\tree$ and $\rec$  stops at $\tpath_0$:
\[\begin{array}{lcl}
\ltree,\rec,\tpath\models \phi& \mbox{ if } & \ltree,\rec,\tpath_{0}\models\phi \\ 
\ltree,\rec,\tpath\models \neg \psi & \mbox{ if } & \ltree,\rec,\tpath\not\models \psi \\ 
\ltree,\rec,\tpath\models \psi \ou \psi' & \mbox{ if } & \ltree,\rec,\tpath \models \psi \mbox{ or }\ltree,\rec,\tpath\models \psi' \\ 
\ltree,\rec,\tpath\models \X\psi & \mbox{ if } & \ltree,\rec,\tpath_{\geq 1}\models \psi \\ 
\ltree,\rec,\tpath\models \psi\until\psi' & \mbox{ if } & \exists\, i\geq 0 \mbox{ s.t.    }\ltree,\rec,\tpath_{\geq
                                                          i}\models\psi' \text{ and } \\
    & & \forall j \text{ s.t. }0\leq j <i, \;\ltree,\rec,\tpath_{\geq j}\models\psi
  \end{array}
  \]

We let $\ltree,\rec\models\phi$ denote $\ltree,\rec,\racine\models\phi$,
where $\racine$ is  $\ltree$'s root.


\vspace{1ex}
\head{Tree unfoldings $\unfold{\state}$}
Let $\CKS=(\setstates,\relation,\lab,\stateinit,\cobsvecinit)$ be a compound Kripke structure over $\APf$. 
The \emph{tree-unfolding of $\CKS$} is the $(\APf,\setstates)$-tree $\unfold{\state}\egdef (\tree,\lab')$, where
    $\tree$ is the set
    of all finite  paths that start in $\stateinit$, and
    for every $\noeud\in\tree$,
    $\lab'(\noeud)\egdef \lab(\last(\noeud))$.
    Given a \CKS $\KS$ and an \EQCTLsi formula $\phi$, we write
    $\KS \models \phi$ if $\unfold[\KS]{\state_\init}, \eobsrecvec\models \phi$,
    where $\eobsrecvec=(0,\cobsvecinit_\ag)_{\ag\in\Agf}$.

\vspace{1ex}
\head{Model-checking problem for \EQCTLsi}
The \emph{model-checking problem for \EQCTLsi} is the 
following: given an instance $(\CKS,\phi)$ where  $\CKS$ is a \CKS and   $\phi$ is an \EQCTLsi
formula, return `Yes' if $\CKS \models \phi$ and `No' otherwise.


Clearly, \EQCTLsi  subsumes
\QCTLsi. Since the latter has an undecidable
model-checking problem~\cite{BMMRV17}, the following is immediate:

\begin{theorem}
    \label{theo-undecidable}
Model checking  \EQCTLsi is undecidable.
\end{theorem}

We now present the syntactic fragment for which we prove that model
checking is decidable. First we adapt the notion of free epistemic formula to the context of \EQCTLsi.
Intuitively, an epistemic subformula $\phi$ of a formula $\Phi$ is free if it does not contain a
free occurrence of a proposition quantified in $\Phi$. To see the
connection with the corresponding notion for \ESL, consider that quantification
on propositions will be used to capture quantification on strategies. 

\begin{definition}
  \label{def-free-epistemice-eqctlsi}
  Let $\Phi\in \EQCTLsi$, and recall that we assume
  $\APq(\Phi)\inter\APfree(\Phi)=\emptyset$. An epistemic subformula $\phi=\Ka\phi'$ of $\Phi$
  is \emph{free in $\Phi$} if $\APq(\Phi)\inter\APfree(\phi)=\emptyset$.
\end{definition}

\bmchanged{For instance, if $\Phi=\existsp{\cobs}(\Ka p) \et \Ka q$, then
subformula $\Ka q$ is free in $\Phi$, but subformula $\Ka p$ is not
because $p$ is quantified in $\Phi$ and  appears free
in $\Ka p$.}



\begin{definition}[Hierarchical formulas]
  \label{def-hierarchical-bis}
  An \EQCTLsi formula $\Phi$ is \emph{hierarchical} if all its epistemic
  subformulas are free in $\Phi$, and for all
  subformulas $\phi_{1},\phi_{2}$ of the form
  $\phi_{1}=\existsp[p_{1}]{\cobs_{1}}\phi'_{1}$ and
  $\phi_{2}=\existsp[p_{2}]{\cobs_{2}}\phi'_{2}$ where
  $\phi_{2}$
  is a subformula of $\phi'_{1}$, we have $\cobs_{1}\subseteq\cobs_{2}$.
\end{definition}

In other words, a formula is hierarchical if epistemic subformulas are
free, and innermost propositional quantifiers
 observe at least as much as
outermost ones. Note that this is very close to
hierarchical formulas of \ESL.
We let \EQCTLsih be the set of hierarchical \EQCTLsi formulas.

\begin{theorem}
  \label{theo-decidable-EQCTLi}
The model-checking problem for \EQCTLsih  is non-elementary decidable.
\end{theorem}

\begin{proof}[Proof sketch]
  We build upon the tree automata construction for \QCTLsi
  presented in~\cite{BMMRV17}, which we extend  to take into account knowledge
  operators and observation change. To do so we resort to the \ktrees
  machinery developed in~\cite{DBLP:journals/iandc/Meyden98,DBLP:conf/fsttcs/MeydenS99}, and
  extended in~\cite{barriere:2018} to the case of dynamic observation
  change. One also needs to observe that free epistemic subformulas can be
  evaluated indifferently in any node of the input tree. 
\end{proof}

\section{Model-checking hierarchical  {\ESL}}
\label{sec-modelcheckingESL}

In this section we prove that  model checking hierarchical instances of \ESL
is decidable
(Theorem~\ref{theo-ESL}), by reduction to the model-checking problem
for \QCTLsih.

Let $(\Phi,\CGSi)$ be a hierarchical instance of the \ESL
model-checking problem. 
The construction of the \CKS is the same as in~\cite{BMMRV17}, except
that in addition we have to deal with initial observations.

\vspace{1ex}
\head{Constructing the \CKS $\CKS_{\CGSi}$} 
 Let $\CGSi=(\Act,\setpos,\trans,\val,\obsint,\posinit,\obsinit)$ and
 $\Obsf=\{\obs_{1},\ldots,\obs_{n}\}$.  For $i \in [n]$, define the
 local states
 $\setlstates_{i}\egdef\{\eqc{\obs_{i}}\mid\pos\in\setpos\}$ where
 $\eqc{\obs}$ is the equivalence class of $\pos$ for relation
 $\obseq$. We also let $\setlstates_{n+1}\egdef\setpos$. Finally, let
 $\APv\egdef\{p_{\pos}\mid\pos\in\setpos\}$ be a set of fresh atomic
 propositions, disjoint from $\APf$.

 Define the \CKS $\CKS_{\CGSi}\egdef(\setstates,\relation,\lab',\stateinit,\cobsvecinit)$ where
\begin{itemize}
\item $\setstates\egdef\{\state_{\pos} \mid \pos\in\setpos\}$,
  
where $\state_{\pos}\egdef(\eqc{\obs_{1}},\ldots,\eqc{\obs_{n}},\pos) \in \prod_{i\in [n+1]}\setlstates_{i}$.
\item $\relation\egdef\{(\state_{\pos},\state_{\pos'})\mid
  \exists\jmov\in\Mov^{\Agf} \mbox{ s.t. }\trans(\pos,\jmov)=\pos'\}
  \subseteq \setstates^2$,
\item $\lab'(\state_{\pos})\egdef\val(\pos)\union \{p_{\pos}\}
  \subseteq \APf \cup \APv$,
\item $\state_{\init}\egdef\state_{\pos_{\init}}$,
\item $\cobsvecinit$ is such that $\cobsvecinit_\ag=\{i\}$ if $\obsinit_\ag=\obs_i$.
\end{itemize}


For every $\fplay=\pos_{0}\ldots\pos_{k}$, 
we let $\noeud_{\fplay}\egdef \state_{\pos_{0}}\ldots
\state_{\pos_{k}}$.  The mapping
$\fplay\mapsto\noeud_{\fplay}$ is a bijection between 
 finite plays in $\CGSi$ and 
nodes in $\unfold[\CKS_{\CGSi}]{\state_{\init}}$.
For
 $i\in[n]$ we let $\cobs_i=\{i\}$, and 
for an observation record $\rec$ in $\CGSi$ 
we let $\rec'$ be the observation record in $\CKS_\CGSi$ where each
 $\obs_i$ is replaced with $\cobs_i$.

\vspace{1ex}
\head{Constructing the \EQCTLsih formulas $\tr[f]{\phi}$} 
Suppose that
$\Mov=\{\mov_{1},\ldots,\mov_{\maxmov}\}$; let  $\APm\egdef\{p_{\mov}^{\var}\mid\mov\in\setmoves
\mbox{ and }\var \in \Varf\}$ be a set of propositions disjoint from $\APf\union\APv$.
For every partial
 function $f:\Agf \partialto \Varf$
we define $\tr[f]{\phi}$
by induction on $\phi$. All cases for boolean, temporal and knowledge
operators are obtained by simply distributing
over the operators of the logic; for instance, $\tr[f]{p}=p$ and $\tr[f]{\Ka\phi}=
\Ka\tr[f]{\phi}$. We now describe the translation for the remaining cases.
\begin{align*}
\tr[f]{\Estrato{\obs}\phi}	& \egdef  \exists^{\trobs{\obs}}
                                  p_{\mov_{1}}^{\var}\ldots
                                  \exists^{\trobs{\obs}}
                                  p_{\mov_{\maxmov}}^{\var}.\, \phistrat
                                  \et \tr[f]{\phi} 
\end{align*}
where $\trobs{\obs_i} = \{j\mid \obsint(\obs_{i})\subseteq\obsint(\obs_{j})\}$, and
\[
\phistrat = \A\G\bigou_{\mov\in\Mov}(p_{\mov}^{\var}\et\biget_{\mov'\neq\mov}\neg p_{\mov'}^{\var}).
\]


Note that  $\cobs_i=\{i\}\subseteq\trobs{\obs_i}$. The definition of $\trobs{\obs_i}$ is tailored to obtain a
hierarchical \EQCTLsi formula. It is correct because for each additional
component in $\trobs{\obs_i}$ (\ie, each $j\neq i$), we have
$\obsint(\obs_i)\subseteq \obsint(\obs_j)$, meaning that each such
component $j$ brings less information than component $i$. A strategy
thus has no more information with $\trobs{\obs_i}$ than it would with $\cobs_i$.




For the binding operator, agent $a$'s observation
becomes the one associated with  strategy variable $x$ (see Remark~\ref{rem-obs-strat}):
\begin{align*}
  \tr[f]{\bind{\var}\phi} & \egdef \Da{\cobs_i}\tr[{f[\ag\mapsto
                            \var]}]{\phi}  \quad  \mbox{if } \obs_\var=\obs_i.
\end{align*}

For the outcome quantifier, we let
\begin{multline*}
  \tr[f]{\Aout\psi} \egdef \A(\psiout[f] \rightarrow \tr[f]{\psi}),\quad\mbox{where}\\
  \psiout[f]= \G
  \bigwedge_{\pos\in\setpos} \Big ( p_{\pos} \rightarrow \bigvee_{\jmov\in\Act^{\Agf}} 
 ( \bigwedge_{\ag\in\dom(f)}p_{\jmov_{\ag}}^{f(\ag)} \wedge \X
p_{\trans(\pos,\jmov)}  )  \Big )
\end{multline*}

The formula $\psiout[f]$ selects paths in which agents who are
assigned to a strategy follow it.

 \begin{proposition}
   \label{prop-redux}
   Suppose that $\free(\phi)\inter\Agf\subseteq\dom(f)$, that for all
   $a \in \dom(f)$, $f(a) =
   x$ iff $\assign(a) = \assign(x)$, and that $\rec$ stops at $\fplay$. Then
\[\CGSi,\assign,\rec,{\fplay}\models\phi\quad \mbox{ if and only if }\quad
\unfold[\CKS_{\CGSi}]{},\rec',\noeud_{\fplay} \models \tr[f]{\phi}.\]
 \end{proposition}

 Applying this to sentence $\Phi$, any assignment $\assign$, 
 $f=\emptyset$,  $\fplay=\pos_\init$ and initial
 observation records, we get that
 \[\CGSi \models \Phi \quad \mbox{if and only if}\quad
\unfold[\CKS_{\CGSi}]{} \models
 \tr[\emptyset]{\Phi}.\]


\vspace{1ex}
\head{Preserving hierarchy}
To complete the proof of Theorem~\ref{theo-ESL} we show that $\tr[\emptyset]{\Phi}$ is a
hierarchical \EQCTLsi formula. 


First, observe that if $\Ka\phi$ is a free epistemic formula in $\Phi$, then
its translation is also a  free epistemic formula in
$\tr[\emptyset]{\Phi}$. Indeed, the only atomic propositions that are
quantified in $\tr[\emptyset]{\Phi}$ are of the form
$p_{\mov}^{\var}$. They code for strategies, and
appear only in translations of strategy quantifiers, where
they are quantified upon, and outcome
quantifiers. Thus they can only appear free
in the translation of  an epistemic formula $\Ka\phi$ if $\phi$
contains an outcome quantifier where some agent uses a strategy that
is not quantified within $\phi$.
Concerning the hierarchy on observations of quantifiers, simply observe that $\Phi$ is hierarchical in
$\CGSi$, and  for every two observations $\obs_{i}$ and $\obs_{j}$ in $\Obsf$ such that
$\obsint(\obs_{i})\subseteq\obsint(\obs_{j})$, by definition of $\trobs{\obs_{k}}$
we have that $\trobs{\obs_{i}}\subseteq \trobs{\obs_{j}}$.

\section{Applications}
\label{sec-applications}

\ESL being very expressive,  many strategic problems with epistemic
temporal specifications can be cast as model-checking problems for
\ESL. Our main result thus provides a decision procedure for such
problems on systems with hierarchical information. We
present two such applications.

\subsection{Distributed synthesis}
\label{sec-dist-synth}

We consider the problem of distributed synthesis from epistemic
temporal specifications studied
in~\cite{van1998synthesis,DBLP:conf/concur/MeydenW05} and we give a  precise definition to its
variant with uninformed semantics of knowledge, discussed
in~\cite{DBLP:conf/mfcs/Puchala10}.

Assume that
$\Agf=\{\ag_1,\ldots,\ag_n,e\}$, where $e$ is a special player called
the \emph{environment}. Assume also that to each player $\ag_i$ is assigned an
observation symbol $\obs_i$. 
The above-mentioned works
 consider
specifications from linear-time epistemic temporal logic \LTLK, which
extends \LTL with knowledge operators. The semantics of
 knowledge operators contains an implicit universal quantification on
 continuations of indistinguishable finite
 plays. In~\cite{van1998synthesis,DBLP:conf/concur/MeydenW05}, which
 considers the informed semantics of knowledge, \ie, where all players
 know each other's strategy, this quantification is restricted to continuations
 that follow these strategies; in~\cite{DBLP:conf/mfcs/Puchala10}, which
 considers the uninformed semantics,  it quantifies over all possible
 continuations in the game.

We now prove a stronger result than the one announced
in~\cite{DBLP:conf/mfcs/Puchala10}, by allowing the use of either
existential or universal quantification on possible continuations
after a knowledge operator.
For an \ESL path formula $\psi$, we define
\begin{multline*}
\Phisynth\egdef
\Estrato[\var_{1}]{\obs_1}\ldots\Estrato[\var_{n}]{\obs_n}
\bind[\ag_1]{\var_1}\ldots\bind[\ag_n]{\var_n}\A\psi.  
\end{multline*}

Note that the outcome quantifier $\A$ quantifies on all possible
behaviours of the environment.
\begin{definition}
  \label{def-uninformed-dist-synth}
  The  epistemic distributed synthesis problem  with uninformed semantics is the following: given a \CGSi $\CGSi$ and an \ESL
  path formula $\psi$, decide whether
$\CGSi\models\Phisynth$.
\end{definition}

 Let \LTLK(\CTLsK) the set of path formulas obtained
by allowing in \LTL subformulas of the form $\Ka\phi$, with
$\phi\in\CTLsK$. The path quantifier from \CTLsK quantifies on all
possible futures, and is simulated in \ESL by an unbinding for all players
followed by an outcome quantifier.
Therefore with specifications $\psi$ in \LTLK(\CTLsK), all epistemic
subformulas are free. It follows that if the system $\CGSi$ is hierarchical, and we assume without loss of
generality that $\obsint(\obs_i)\supseteq \obsint(\obs_{i+1})$, then
$(\CGSi,\Phisynth)$ is a hierarchical instance of \ESL.

\begin{theorem}
  \label{theo-uninformed-dist-synth}
  The epistemic distributed synthesis problem from specifications in \LTLK(\CTLsK)
 with uninformed semantics is decidable on
  hierarchical systems.
\end{theorem}

In fact we can deal with even richer specifications: as long as
hierarchy is not broken and epistemic subformulas remain free, it is
possible to re-quantify on agents' strategies inside an epistemic
formula. Take for instance formula
\[\psi=\F\Ka[\ag_n]\bind[\ag_1]{\unbind}\ldots\bind[\ag_n]{\unbind}\Estrato{\obs_n}\bind[\ag_n]{\var}\A\G\Ka[\ag_n]
  p\]
It says that
eventually, agent $a_n$ knows that she can change strategy 
 so that in all outcomes of this strategy, she will
always know that $p$ holds. If $\CGSi$ is hierarchical, then
$\Phisynth$ forms a hierarchical instance with $\CGSi$. Consider now formula
\[\psi=\F\Ka[\ag_i]\bind[\ag_1]{\unbind}\ldots\bind[\ag_n]{\unbind}\Astrato{\obs}\bind[\ag_j]{\var}\E\G\neg\Ka[\ag_j]
  p\]
which means that eventually agent $\ag_i$ knows that for any strategy
with observation $\obs$ that  agent $\ag_j$ may take, 
there is an outcome in which $\ag_j$ never  knows that $p$
holds. If $\obs$ is finer than $\obs_n$ (and thus all other $\obs_i$),
for instance if $\obs$ represents perfect information, then
hierarchy is preserved and we can solve distributed synthesis for this
specification. In addition, the semantics of our knowledge operator
takes into account the fact that agent $\ag_j$ changes observation power.


\subsection{Rational synthesis}
\label{sec-ratsyn}


Consider $\Agf=\{\ag_1,\ldots,\ag_n\}$, each player $\ag_i$ having
observation symbol $\obs_i$.
Given 
a global objective $\psi_g$ and  individual objectives
$\psi_i$ for each player $\ag_i$, define
\begin{multline*}
\Phirat\egdef  \Estrato[\var_{1}]{\obs_1}\ldots\Estrato[\var_{n}]{\obs_n}\bind[\ag_1]{\var_1}\ldots\bind[\ag_n]{\var_n}\A\psi_g\\
  \wedge
   \bigwedge_{i\in\{1,\ldots,n\}}\Big (\Astrato[\var_i]{\obs_i}\left(\bind[\ag_i]{\var_i}\,\A\psi_i\right
   )\impl   \A\psi_i\Big).
\end{multline*}
It is easy to see that $\Phirat$ expresses the existence of a solution
to the cooperative rational synthesis
problem~\cite{DBLP:journals/amai/KupfermanPV16,DBLP:conf/icalp/ConduracheFGR16}. However
this formula does not form hierarchical instances, even with
hierarchical systems. 
But the same argument used
in~\cite{BMMRV17} for Nash equilibria shows that $\Phirat$ is 
 equivalent to $\Phirat'$,  obtained from
$\Phirat$ by replacing each $\Astrato[\var_i]{\obs_i}$ with $\Astrato[\var_i]{\obs_p}$,
where $\obs_p$ represents perfect observation.

\begin{theorem}
  \label{theo-rat}
  Rational synthesis from \LTL(\CTLsK) specifications is decidable on
  hierarchical systems.
\end{theorem}

As in the case of distributed synthesis discussed before, we can in
fact handle more complex specifications, nesting knowledge and
strategy quantification.

\section{Discussion}
\label{discussion}

  In the uninformed semantics, 
players ignore each other's strategy, but they also ignore their 
 own, in the sense that they consider possible finite plays in
 which they act differently from what their strategy prescribes.
This is the usual semantics in epistemic strategic
  logics~\cite{Hoek03ATELstudialogica,DBLP:journals/fuin/JamrogaH04,epistemicSL,dima2010model,DBLP:conf/ijcai/BelardinelliLMR17,BelardinelliLMR17a},
  and in some situations it may be what one wants to model. For
  instance, an agent may  execute her strategy step
  by step without having access to what her strategy
  would have prescribed in alternative plays. In this case, 
  where the agent cannot know whether a possible play follows her
  strategy or not,  the uninformed semantics of knowledge is
  the right one.

  On the other hand it seems natural, especially formulated in these
  terms, to assume that an agent knows her own strategy. We describe
  how,  in the case where agents do not change strategies or
  observation along time, this semantics can be retrieved within the
  uninformed semantics.

Assume that 
  player $\ag$ is assigned some observation symbol $\obs_\ag$.  
  As pointed
  out in~\cite[p.16]{DBLP:conf/mfcs/Puchala10}, in the setting of
  synchronous perfect recall, letting a player know her strategy
  is equivalent to letting her remember her own actions.  To see this,
  assume 
 that
  finite plays also contain each joint action between two positions,
  and let $\obseq[\obs_\ag]'$ be such that $\pos_\init
  \jmov_1\pos_1\ldots\jmov_n\pos_n\obseq[\obs_\ag]'\pos_\init
  \jmov_1\pos_1\ldots\jmov_n\pos_n$ if for all $i\in\{1,\ldots,n\}$,
  $\jmov_a=\jmov'_a$ and $\pos_i\obseq[\obs_\ag]\pos'_i$.   Then, for a
  strategy $\strat$ of player $a$ and two finite plays
  $\fplay,\fplay'$ such that $\fplay\obseq[\obs_\ag]'\fplay'$, it
  holds that $\fplay$ is consistent with $a$ playing $\strat$ iff 
  $\fplay'$ is consistent with $a$ playing $\strat$. This is because
for every $i<n$ we have $\fplay_{\leq i}\obseq[\obs_\ag]'\fplay'_{\leq
  i}$ (perfect recall),
the next action taken by player $a$ is the same after $\fplay'_{\leq
  i}$ and $\fplay'_{\leq i}$ (definition of $\obseq[\obs_\ag]'$), and $\strat$ being an
$\obs_\ag$-strategy it is defined similarly on
both prefixes.

In our setting, moves are not part of finite plays. To simulate the
relation $\obseq[\obs_\ag]'$ in which agent $\ag$ remembers her own actions,
one can put inside the positions of game structures the information
of the last joint move played, possibly duplicating some positions. 
One then refines each
observation $\obs_\ag$ to only consider two positions equivalent if
they contain the same move for player $a$. We then get a semantics
where each agent remembers her own actions which, if agents do not
change strategy or observation through time, is equivalent to
knowing her own strategy. Note that doing so, a system can only be
hierarchical if more informed players also observe all actions of less
informed ones.

In the general case, where players can change strategies and observations, we
do not know to what extent we can deal with the variant of the
uninformed semantics where players know there own strategies. We leave
this for future work.

\section{Conclusion}
\label{sec-conclusion}

In this paper we have discussed two possible semantics of knowledge
when combined with strategies, the informed and uninformed
one. Focusing on the latter, we introduced \ESL, a very expressive
logic to reason about knowledge and strategies in distributed systems
which can handle sophisticated epistemic variants of game-theoretic
notions such as Nash equilibria.
In addition, it is the first  logic of knowledge and strategies that
permits reasoning about agents whose observation power may change.
This is a very natural phenomenon: one may think of a program that receives
access to previously hidden variables, or a robot that loses a sensor.

We solved the model-checking problem
of our logic for hierarchical instances. To do so, we introduced an
extension of \QCTLs with epistemic operators and  operators of
observation change, and we developed an automata construction based
on tree automata and $k$-trees.
This is the first decidability result for a logic of
strategies, knowledge and time with perfect recall on systems with
hierarchical information. Besides, it is also the first result for
epistemic strategic logics that takes into account dynamic changes of observation.

Our result implies that   distributed synthesis and
rational synthesis for epistemic
temporal specifications and  the
uninformed semantics of knowledge are decidable on hierarchical
systems.
Similar results
for other solution concepts, such as subgame-perfect equilibria or
admissible strategies~\cite{DBLP:conf/concur/BrenguierPRS17}, could
be
obtained similarly.

\section{Acknowledgements}
This project has received funding from the European Union's Horizon 2020 research  and innovation programme under the Marie Sklodowska-Curie grant agreement No 709188.
\UElogo

\bibliography{biblio}
\bibliographystyle{aaai}

\end{document}